\documentclass[reqno,12pt,centertags]{amsart}
\usepackage{amsmath,amsthm,amscd,amssymb,latexsym,upref,stmaryrd,epsfig,graphicx}
\usepackage[cp1251]{inputenc}
\usepackage[english]{babel}
\usepackage{mathtext}
\usepackage{amsfonts}
\usepackage{graphicx}
\usepackage{epsfig}

\usepackage{hyperref}
\usepackage[numbers,sort&compress]{natbib}
\newcommand*{\mailto}[1]{\href{mailto:#1}{\nolinkurl{#1}}}
\usepackage{amsmath}
\usepackage{amsthm}

\newtheorem{theorem}{Theorem}[section]
\newtheorem{lemma}{Lemma}[section]
\newtheorem{remark}{Remark}[section]
\newtheorem{prop}{Proposition}[section]

\newtheorem{algorithm}{Algorithm}[section]
\newtheorem{definition}{Definition}[section]
\numberwithin{equation}{section}

\begin{document}

\thispagestyle{empty}

\noindent{\large\bf Solvability of the inverse scattering problem for the self-adjoint matrix Schr\"{o}dinger operator on the half line}
\\

\noindent {\bf  Xiao-Chuan Xu}\footnote{Department of Applied
Mathematics, School of Science, Nanjing University of Science and Technology, Nanjing, 210094, Jiangsu,
People's Republic of China, {\it Email:
xiaochuanxu@126.com}}
{\bf and Chuan-Fu Yang}\footnote{Department of Applied
Mathematics, School of Science, Nanjing University of Science and Technology, Nanjing, 210094, Jiangsu,
People's Republic of China, {\it Email: chuanfuyang@njust.edu.cn}}
\\

\noindent{\bf Abstract.}
{In this work we study the inverse scattering problem for the self-adjoint matrix Schr\"{o}dinger operator on the half line. We provide the necessary and sufficient conditions for the solvability of the inverse scattering problem.}



\section{Introduction}
Consider the matrix Schr\"{o}dinger equation on the half-line
\begin{equation}\label{1}
-\Psi''(x)+V(x)\Psi(x)=k^2\Psi(x),\quad x\in \mathbb{R}^+:=(0,+\infty),
\end{equation}
where $\Psi$ is either an $n\times n$ matrix-valued function or a column vector-valued function with $n$ components, and
the potential $V(x)$ belongs to $L_1^1(\mathbb{R}^+)$ and satisfies
\begin{equation}\label{2}
V(x)^\dag=V(x).
\end{equation}
Here the dagger "$\dag$" denotes the matrix adjoint (complex conjugate and matrix transpose), and $V\in L_1^1(\mathbb{R}^+)$ means that
\begin{equation*}
  \int_0^\infty (1+x)\|V(x)\|dx<\infty,
\end{equation*}
where the matrix norm, without loss of generality, is defined as
\begin{equation*}
  \|V(x)\|=\max_l\sum_{s=1}^n|V_{ls}(x)|,
\end{equation*}
here $V_{ls}(x)$ denotes the $(l, s)$-entry of $V(x)$.

Let us consider (\ref{1}) with the most general self-adjoint boundary condition \cite{AKW,AKW1,AW,W,MH1,MH2,MH3}
\begin{equation}\label{3}
 -B^\dag \Psi(0)+A^\dag \Psi'(0)=0,
\end{equation}
where
\begin{equation}\label{4}
A=\frac{1}{2}\left(U+I_n\right)\quad\text{and}\quad B=\frac{i}{2}\left(U-I_n\right)
\end{equation}
for some constant $n \times n$ unitary matrix $U$. Here $I_n$ denotes the  $n \times n$ identity matrix. It is obvious that the matrices $A$ and $B$ satisfy
\begin{equation}\label{4a}
 B^\dag A=A^\dag B,\quad A^\dag A+B^\dag B>0.
\end{equation}
It should be pointed out that the conditions (\ref{4}) and (\ref{4a}) are equivalent. Indeed, it was shown in \cite{AKW} that if $A$ and $B$ satisfy (\ref{4a}) then there exists a unitary matrix $U$ and invertible matrix $D$ such that $AD$ and $BD$ have the form of (\ref{4}). Here, for convenience, we use (\ref{4}) in our paper.
Denote by $L(V,U)$ the problem (\ref{1}) and (\ref{3}).

Matrix Schr\"{o}dinger equation arises in quantum mechanics, electronics, nano science and other branches of science and engineering \cite{ZS,KS,KS1,KS2,PK,PK1}. The matrix Schr\"{o}dinger equation (\ref{1}) with self-adjoint boundary condition (\ref{3}) is connected with scattering in quantum mechanics involving particles of internal structures as spins, scattering on graphs and quantum wires (see \cite{BCFK,SBF,MH1,MH2,KS,KS1,KS2,EKKST,PK,PK1} and the references therein).

Inverse scattering problem for the scalar Schr\"{o}dinger operator on the half line has been studied thoroughly \cite{BL,VM,AR}. Levitan \cite{BL}, Marchenko \cite{VM}, and Ramm \cite{AR} respectively applied three different kinds of methods to obtain the characterization of the scattering data.

Matrix Schr\"{o}dinger operators are more complicated than the scalar ones. Some aspects of the matrix cases on a finite interval were studied in \cite{RC,CY,VY}, and the matrix cases on the full line were studied in \cite{AKW0,NB,GKM,WK}. For the matrix Schr\"{o}dinger operator on the half line, there are also some results \cite{AM,AKW,AKW1,AW,NB2,MH1,MH2,MH3}, whereas, along with many unsolved problems. Specifically,
Agranovich and Marchenko \cite{AM} investigated the matrix Schr\"{o}dinger operator with Dirichlet boundary condition, i.e. the problem $L(V,-I_n)$, and gave the necessary and sufficient conditions for the solvability of the inverse scattering problem. Harmer \cite{MH1,MH2,MH3} considered the problem $L(V,U)$, deduced the Marchenko equation, and studied its applications in the scattering problem on graphs. Aktosun et al \cite{AKW,AKW1,AW} investigated the small and high energy asymptotic behaviors for the scattering data, and also obtained Levinson's theorem. We also mention that some aspects of non-self-adjoint matrix differential operators are studied in \cite{NB1,NB3,FY}

As far as we know, there is no result dealing with the uniqueness and solvability of the inverse scattering problem for the problem $L(V,U)$, which are the main problems we shall study in this paper. The most general self-adjoint matrix Schr\"{o}dinger operator $L(V,U)$ is more complicated than the Dirichlet one, especially in the study of the inverse scattering problem. Because the unitary matrix $U$ in the boundary condition can yield many  possibilities. In the study of inverse scattering problem for the Dirichlet case \cite{AM}, the authors just assumed that the homogeneous equation corresponding to the Marchenko equation has only zero solution, which yields the unique solvability of the Marchenko equation by the Freldholm alternative theorem.  In this paper, we shall provide the more direct conditions on the scattering data that guarantee the unique solvability of the Marchenko equation.
In addition, we will also give the condition which guarantees that the recovered matrix $U$ in the boundary condition is unitary. In the Dirichlet case, there is no need to study it, because the boundary condition is known a priori. However in the most general self-adjoint case, it has to be considered.
It should be explained that
without the small and high energy asymptotic behaviors of the scattering matrix and the Jost matrix in \cite{AKW,AW} by Aktosun et al, one can't rigorously deduce the the Marchenko equation. In this paper, with the help of the results in \cite{AKW,AW}, we rigorously deduce the Marchenko equation.

This paper is organized as follows. Section 2 deals with some notations and known results in \cite{AM,AKW0,AKW,AW}, which are used in the next sections. In Section 3, we deduce the Marchenko equation, and introduce the so-called normalization matrices. Section 4 is concerned with the behaviors of the matrix-valued function $F(x)$, which is determined by the scattering data. In Section 5, we study the self-ajointness of $V(x)$ and unitarity of $U$ recovered from the scattering data. Section 6 presents the solvability conditions on the scattering data for solving the Marchenko equation.  The last section gives the uniqueness theorem of the inverse scattering problem, and the necessary and sufficient conditions for the scattering data.

\section{Preliminaries}

In this section, we introduce some notations, and recall the definitions of some relevant physical qualities together with their properties. The results presented in this section can be found in  \cite{AM,AKW0,AKW,AKW1,AW}.

Let $\mathbb{C}^n$ ($\mathbb{C}^{n,T}$) be the space of complex row (column) vector with $n$ components. The norm and the inner product are defined as follows.
\begin{equation*}
  \|x\|=\sum_{j=1}^n|x_j|,\quad(x,y):=xy^\dag=\sum_{j=1}^n{x}_j\bar{{y}}_j\;(\text{or}\;\sum_{j=1}^n\bar{x}_j{y}_j),\quad \forall x,y\in\mathbb{ C}^n\; (\text{or}\;\mathbb{C}^{n,T}),
\end{equation*}
where $\bar{y}_j$ means the complex conjugate of the number $y_j$. Denote $\mathbb{C}^\pm=\{k:\pm{\rm Im}k>0\}$ and $\overline{\mathbb{C}}^\pm=\mathbb{C}^\pm\cup\mathbb{R}$, and $\mathbb{R}^+=(0,+\infty)$. Let $L^p((a,b);\mathbb{C}^n)$, $L^p((a,b);\mathbb{C}^{n,T})$ and $L^p((a,b);\mathbb{C}^{n\times n})$ be the spaces of all row vector-, column vector- and matrix-valued functions, respectively, with each  element belonging to $L^p(a,b)$, where $p=1,2$ and $-\infty\le a<b\le+\infty$. If it does not cause misunderstanding, we may just use the notation $L^p(a,b)$.
We define the norms as
\begin{equation*}\label{y1}
\!\!\!\|g\|_{L^p}\!:=\sum_{j=1}^n\|g_j\|_{L^p},\;\; g\in L^p((a,b);\mathbb{C}^n)\;(\text{or }L^p((a,b);\mathbb{C}^{n,T})),\;\; p=1,2,
\end{equation*}
where $g_j(x)$ is the $j$-th component of $g(x)$. The notation "$0_{n}$" used below means an $n\times n$ zero matrix, sometimes, we also use "$0$" instead of an $n\times m$ zero matrix.
An overdot used below means the $k$-derivative.
By a generic constant $c$, we mean a constant that does not necessarily take the same value in each appearance.

Together with (\ref{1}), we consider the following equation
\begin{equation}\label{5}
-\Phi''+\Phi V(x)=k^2\Phi,\quad x>0,
\end{equation}
where $\Phi$ is either an $n\times n$ matrix-valued function or a row vector-valued function with $n$ components.
Let $[\Phi;\Psi]:=\Phi\Psi'-\Phi'\Psi$ denote the Wronskian. It is easy to prove that the Wronskian $[\Phi(k,x);\Psi(k,x)]$ does not depend on $x$.
In addition, if $\Psi(k,x)$ is a solution to (\ref{1}) for real $k$, then $\Psi(\pm k,x)^\dagger$ are solutions to (\ref{5}). In case $\Psi(k,x)$ and $\Psi(-\bar{k},x)^\dagger$ have analytic extensions from $k\in\mathbb{R}$ to $k\in\mathbb{C}^+$ and $\Psi(k,x)$ solves (\ref{1}), then $\Psi(-\bar{k},x)^\dagger$ solves (\ref{5}), and hence $[\Psi(-\bar{k},x)^\dag;\Psi(k,x)]$ is independent of $x$ and analytic for $k\in\mathbb{C}^+$.

Let $f(k,x)$ be the Jost solution to (\ref{1}), which is a matrix-valued function satisfying the integral equation
\begin{equation}\label{jost}
  f(k,x)=e^{ikx}+\int_x^\infty\frac{\sin k(t-x)}{k}V(t)f(k,t)dt,\quad k\in \overline{{\mathbb{C}}}^+.
\end{equation}

Let us recall some properties of the Jost solution $f(k,x)$ \cite{AM,AKW0,AKW}.
\begin{prop}\label{2.1}
Assume $V\in L_1^1(\mathbb{R}^+)$ and satisfies (\ref{2}). Then the following assertions hold.

(a) For each fixed $x\ge0$, the matrix-valued functions $f^{(v)}(k,x)$ ($v=0,1$) are analytic for $k\in\mathbb{C}^+$ and continuous for $k\in \overline{\mathbb{C}}^+$, and satisfies
\begin{equation}\label{an}
\!\!\!f(k,x)\!=e^{ikx}\left[\!I_n\!+\!\frac{Q(x)}{ik}\!+\!o\left(\!\frac{1}{k}\!\right)\!\right]\!,\; Q(x)\!:=\frac{\int_x^\infty V(t)dt}{2},\; |k|\to\infty\;\text{ in } \;\overline{\mathbb{C}}^+.
\end{equation}

(b) For each fixed $k\in \overline{\mathbb{C}}^+\setminus\{0\}$, the Jost solution $f(k,x)$ satisfies
\begin{equation}\label{yyy}
 f^{(v)}(k,x)=(ik)^v e^{ikx}[1+o(1)],\quad x\to+\infty,\quad v=0,1.
\end{equation}

(c)
The Jost solution $f(k,x)$ satisfies the Wronskian relations
\begin{equation}\label{10}
 [f(\pm k, x)^\dagger;f(\pm k,x)]=\pm2ikI_n,\quad k\in \mathbb{R},
\end{equation}
\begin{equation}\label{11}
 [f(- \bar{k}, x)^\dagger;f( k,x)]=0_n,\quad k\in \overline{\mathbb{C}}^+.
\end{equation}

(d) For $k\in \mathbb{R}\setminus\{0\}$, the Jost solutions $f(k,x)$ and $f(-k,x)$ form a fundamental system of solutions to (\ref{1}).

(e)
 The Jost solution $f(k,x)$ can be presented as
\begin{equation}\label{6}
 f(k,x)=e^{ikx}+\int_x^\infty K(x,t)e^{ikt}dt,\quad x\ge0,\quad k\in\overline{\mathbb{C}}^+,
\end{equation}
 where $K(x,t)$ is a continuous function of two variables, having first derivatives with respect to $x$ and $t$, and satisfying
\begin{equation}\label{7}
  \!\!||K(x,t)||\le c \sigma\left(\frac{x+t}{2}\right), \;\;\sigma(x):=\int_x^\infty ||V(t)||dt\in L^p(\mathbb{R}^+),\;\; p=1,2,
\end{equation}
\begin{equation}\label{8}
 V(x)=-2\frac{dK(x,x)}{dx},
\end{equation}
and
\begin{equation}\label{9}
\!\!\left\|K_x(x,t)\!+\!\frac{1}{4}V\left(\frac{x+t}{2}\right)\right\|+\left\|K_t(x,t)\!+\!\frac{1}{4}V\left(\frac{x+t}{2}\right)\right\|\le c\sigma(x) \sigma\!\!\left(\!\frac{x+t}{2}\!\right).
\end{equation}

\end{prop}

Let $\varphi(k,x)$ be the solution to (\ref{1}) satisfying the initial condition
\begin{equation}\label{j3}
  \varphi(k,0)=A,\quad \varphi'(k,0)=B.
\end{equation}
It is known that $\varphi(k,x)$ satisfies the integral equation
\begin{equation*}
\varphi(k,x)=A\cos kx+B\frac{\sin kx}{k}+\int_0^x\frac{\sin k(x-t)}{k}V(t)\varphi(k,t)dt,
\end{equation*}
from which it follows that for each $x\ge0$, the function $\varphi(k,x)$ is analytic in $\mathbb{C}$, and satisfies that as $|k|\to\infty$ in $\mathbb{C}$,
\begin{equation}\label{j2}
\varphi(k,x)=A\cos kx+B\frac{\sin kx}{k}+\int_0^x \frac{\sin k(x-t)\cos kt}{k}V(t)dtA+O\left(\frac{e^{|{\rm Im}k|x}}{k^2}\right).
\end{equation}

Define the Jost matrix \cite{AKW,AKW1,AW,MH1,MH2,MH3}
\begin{equation}\label{12}
J(k):=[f(-\bar{k},x)^\dagger;\varphi(k,x)]=f(-\bar{k},0)^\dagger B-f'(-\bar{k},0)^\dagger A.
\end{equation}

Let us recall that the \emph{eigenvalue} of the problem $L(V,U)$ is the $k^2$-value for which (\ref{1}) has a nonzero column vector solution $\Psi\in L^2(\mathbb{R}^+)$ satisfying (\ref{3}), and the corresponding solution is called an \emph{eigenfunction}. By saying that $k^2$ is an eigenvalue with multiplicity $m_k$, we mean that there are $m_k$ eigenfunctions, which are linearly independent, corresponding to the eigenvalue $k^2$.

Let us summarize the relations between the eigenvalues and the Jost matrix in the following proposition, as well as the asymptotics of the Jost matrix (see \cite{AKW,AW} for details).
\begin{prop}\label{2.2}
Assume $V\in L_1^1(\mathbb{R}^+)$ and satisfies (\ref{2}). Then the Jost matrix $J(k)$ defined in (\ref{12}) is analytic in $\mathbb{C}^+$ and continuous in $\overline{\mathbb{C}}^+$, and satisfies the following conclusions.

(a) $J(k)$ is invertible for $k\in \mathbb{R}\setminus\{0\}$, and it is possible to have $\det J(0)=0$.

(b) The function $\det J(k)$ has a finite number of zeros in $\mathbb{C}^+$, which coincide with the eigenvalues of the problem $L(V,U)$ (including multiplicity), and all of them lie on the positive imaginary axis $i\mathbb{R}^+$;

(c) If $\det J(i\kappa)=0$ with $\kappa>0$ then $J(k)^{-1}$ has a simple pole at $k=i\kappa$, namely,
\begin{equation}\label{13}
 J(k)^{-1}=\frac{N_{-,\kappa}}{k-i\kappa}+N_{0,\kappa}+O(k-i\kappa),\quad k\to i\kappa\; \text{ in }\; \overline{\mathbb{C}}^+,
\end{equation}
\begin{equation}\label{j4}
  J(k)=J(i\kappa)+\dot{J}(i\kappa)(k-i\kappa)+O((k-i\kappa)^2),\quad k\to i\kappa,\; \text{ in }\; \overline{\mathbb{C}}^+,
\end{equation}
where
\begin{equation}\label{j5}
\dot{J}(i\kappa)=\dot{f}'(i\kappa,0)^\dagger A-\dot{f}(i\kappa,0)^\dagger B,
\end{equation}
and the multiplicity of the eigenvalue $-\kappa^2$ of the problem $L(V,U)$  is equal to $\dim \ker(J(i\kappa))$, and
\begin{equation}\label{j6}
A_\kappa:=\int_0^\infty f(i\kappa,x)^\dagger f(i\kappa,x)dx=\frac{\dot{f}'(i\kappa,0)^\dagger f(i\kappa,0)-\dot{f}(i\kappa,0)^\dagger f'(i\kappa,0)}{-2i\kappa};
\end{equation}
(d) As $|k|\to\infty$  in $\overline{\mathbb{C}}^+$, we have
\begin{equation}\label{14}
J(k)^{-1}=O\left(1\right),\quad J(k)=-ik A+B+Q(0)A+o(1).
\end{equation}
\end{prop}

Define the scattering matrix \cite{AKW,AKW1,AW,MH1,MH2,MH3}
\begin{equation}\label{15}
S(k)=-J(-k)J(k)^{-1},\quad k\in \mathbb{R}.
\end{equation}
The properties of the scattering matrix $S(k)$ are summarized as follows \cite{AKW,AW}.
\begin{prop}\label{2.3}
Assume $V\in L_1^1(\mathbb{R}^+)$ and satisfies (\ref{2}). Then the following assertions hold.

(a) The scattering matrix $S(k)$ defined in (\ref{15}) is  continuous on $\mathbb{R}$ and satisfies
\begin{equation}\label{16}
 S(-k)=S(k)^\dagger=S(k)^{-1},\quad k\in\mathbb{R},
\end{equation}
and
\begin{equation}\label{17}
S(k)={U_0}+O\left(\frac{1}{k}\right),\quad |k|\to\infty,
\end{equation}
where $U_0$ is a unitary Hermitian matrix.

(b) If the unitary matrix $U$ appearing in (\ref{4}) has the eigenvalue $-1$ with multiplicity $n_D$, then
the matrix $U_0$ has the eigenvalue $1$ with multiplicity $n-n_D$ and the eigenvalue $-1$ with  multiplicity $n_D$.

(c) As $|k|\to\infty$ in $\overline{\mathbb{C}}^+$, we have
\begin{equation}\label{oa}
AJ(k)^{-1}=-\frac{1}{2ik}[I_n+{U_0}]+O\left(\frac{1}{k^2}\right),\; BJ(k)^{-1}=\frac{1}{2}[I_n-{U_0}]+O\left(\frac{1}{k}\right).
\end{equation}

(d) For some unitary matrix $M$, if $M^\dag UM= -{\rm diag}\{e^{-2i\theta_1},...,e^{-2i\theta_{n-n_D}},I_{n_D}\},$ $\theta_j\in(0,\pi),\;j=\overline{1,n-n_D}$,  then $ U_0=M{\rm diag} \{I_{n-n_D},-I_{n_D}\}M^\dag$. Conversely, if $M^\dag U_0M={\rm diag} \{I_{n-n_D},-I_{n_D}\}$ for some unitary matrix $M$, then $M^\dag U M={\rm diag} \{X_1,-I_{n_D}\}$ for some $(n-n_D)\times (n-n_D)$ invertible matrix $X_1$, and $M^\dag(I_n+U)MT={\rm diag} \{I_{n-n_D},0_{n_D}\}$ for the invertible matrix $T={\rm diag} \{(X_1+I_{n-n_D})^{-1},\frac{1}{2}I_{n_D}\}$.

(e) There exist invertible matrices $T_1, M_1$ such that
\begin{equation}\label{3e}
 M_1 J(k)T_1=
 \begin{bmatrix}
kI_\mu+o(k) & o(k) \\ o(k) & I_{n-\mu}+o(1)
\end{bmatrix}, \quad k\to0,\;\;k\in \overline{\mathbb{C}}^+,
\end{equation}
and
\begin{equation}\label{3e1}
 M_1 S(0) M_1^{-1}=  \begin{bmatrix}
I_\mu & 0 \\0 & -I_{n-\mu}
\end{bmatrix},
\end{equation}
where $\mu\ge0$ is the  multiplicity of the eigenvalue 1 of the matrix $S(0)$.
\end{prop}
\begin{proof}
The assertion (a) has been proved in \cite{AW}. We only show  the assertions (b), (c), (d) and (e).

(b) From Eqs.(4.8), (4.11)$-$(4.13), (5.7), (5.10) and (7.20) in \cite{AW}, we see that if the matrix $(B-iA)(B+iA)^{-1}$ has the eigenvalue $1$ with multiplicity $n_D$, then the limit ${U_0}$ of $S(k)$ as $|k|\to\infty$ on $\mathbb{R}$ has the eigenvalue $-1$ with  multiplicity $n_D$ and the eigenvalue $1$ with  multiplicity $n-n_D$. (One should observe that the notation $U$ used in (4.8) in \cite{AW} is not the same as that in our paper.) Using (\ref{4}) in our paper, one gets $(B-iA)(B+iA)^{-1}=-U^\dagger$. It is obvious that if $1$ is the eigenvalue of $-U^\dagger$ with  multiplicity $n_D$ then $-1$ is the eigenvalue of $U$ with  multiplicity $n_D$. We have proved the assertion (b).

(c) From the formulas (5.7), (5.9), (7.15) (or (7.16)) and (7.20) in \cite{AW}, we obtain the first equation in (\ref{oa}). There is no direct formula in \cite{AW} which can be used to obtain the second equation in (\ref{oa}) like the first one. However, one can use a similar way that the authors of \cite{AW} used in Eq.(5.9) of \cite{AW} to get the asymptotics of $BJ_0(k)^{-1}$. Then one can use the same formulas (7.16) and (7.20) in \cite{AW} to obtain the second equation in (\ref{oa}) here.

(d) Firstly, from (\ref{14}) and (\ref{oa}), we obtain
\begin{equation*}
  A=\frac{1}{2}(I_n+U_0)A+O\left(\frac{1}{k}\right),\;B=-\frac{ik}{2}(I_n-U_0)A+O(1),\quad |k|\to\infty,\;\;k\in \overline{\mathbb{C}}^+,
\end{equation*}
which implies $(I_n-U_0)A=0$, or equivalently, from (\ref{4}) that
\begin{equation}\label{rwq}
(I_n -U_0)(U+I_n)=0.
\end{equation}

From Eqs.(4.8), (4.11) (5.7) and (5.10) in \cite{AW}, together with the proof of (b) above, we know that if $U$ satisfies the assumption in (d), then $M^\dag U_0M $ is a diagonal matrix with diagonal entries either $1$ or $-1$. Furthermore, it follows from (\ref{rwq}) that $M^\dag U_0M ={\rm diag} \{I_{n-n_D},-I_{n_D}\}$.

Conversely, if $M^\dag U_0M={\rm diag} \{I_{n-n_D},-I_{n_D}\}$ for some unitary matrix $M$, then from (\ref{rwq}) we get
\begin{equation}\label{jko}
\begin{bmatrix}
0_{n-n_D}&0_{(n-n_D) \times n_D}\\0_{n_D\times (n-n_D)  }&I_{n_D}
\end{bmatrix}
\begin{bmatrix}
I_{n-n_D}+X_1 &X_2 \\ X_3 & I_{n_D}+X_4
\end{bmatrix}=0,
\end{equation}
where
\begin{equation*}
  \begin{bmatrix}
X_1 &X_2 \\ X_3 &X_4
\end{bmatrix}=M^\dag U M.
\end{equation*}
It follows from (\ref{jko}) that $X_3=0$ and $X_4=-I_{n_D}$.
Since $M^\dag U M$ is a unitary matrix, we have
\begin{equation}\label{unv}
\begin{bmatrix}
X_1 &X_2 \\ X_3& X_4
\end{bmatrix}^{-1}=\begin{bmatrix}
X_1^\dag &X_3^\dag \\ X_2^\dag & X_4^\dag
\end{bmatrix}.
\end{equation}
Let us recall the assertion $(\ast)$: if the matrices $\begin{bmatrix}
X_1 &X_2 \\ X_3 & X_4
\end{bmatrix}$ and $X_4$ are invertible then the matrix $X:=X_1-X_2X_4^{-1}X_3$ is invertible and
\begin{equation}\label{e11}
\begin{bmatrix}\!
X_1 &X_2 \\ X_3 & X_4
\end{bmatrix}^{-1}\!\!\!=\begin{bmatrix}
 X^{-1}&-X^{-1}X_2X_4^{-1} \\ -X_4^{-1}X_3X^{-1} & X_4^{-1}X_3X^{-1}X_2X_4^{-1}+X_4^{-1}\!\end{bmatrix}.
\end{equation}
Using (\ref{e11}), together with $X_3=0$ and $X_4=-I_{n_D}$, we have from (\ref{unv}) that $X_2=0$. Thus, $M^\dag UM={\rm diag} \{X_1,-I_{n_D}\}$. Note that the matrix $I_{n-n_D}+X_1$ must be invertible. Otherwise, $U$ has the eigenvalue $-1$ with multiplicities more than $n_D$, which implies from (b) that $U_0$ has the eigenvalue $-1$ multiplicities also more than $n_D$. This leads to a contradiction with the assumption $M^\dag U_0M={\rm diag} \{I_{n-n_D},-I_{n_D}\}$.
Denote $T:={\rm diag} \{(X_1+I_{n-n_D})^{-1},\frac{1}{2}I_{n_D}\}$. Then we get $M^\dag (I_n+U)MT={\rm diag} \{I_{n-n_D},0_{n_D}\}$.

(e) From Eqs.(6.21)$-$(6.23) in \cite{AKW}, we know that there exist invertible matrices $M_0$ and $T_0$ such that
\begin{equation}\label{3ex}
 M_0 J(k)T_0=
 \begin{bmatrix}
kA_0+o(k) & kB_0+o(k) \\ kC_0+o(k) & D_0+o(1)
\end{bmatrix}, \;\; k\to0,\;\;k\in \overline{\mathbb{C}}^+,
\end{equation}
where $A_0$ and $D_0$ are $\mu\times\mu$ and $(n-\mu)\times(n-\mu)$ invertible constant matrices, respectively, and $B_0$ and $C_0$ are constant matrices. Postmultiplying both sides of (\ref{3ex}) by the matrix $T_2:={\rm diag}\{A_0^{-1}, D_0^{-1}\}$, we obtain
\begin{equation}\label{x.1}
   M_0 J(k)T_0 T_2=
\begin{bmatrix}
kI_\mu+o(k) & kB_0D_0^{-1}+o(k) \\ kC_0A_0^{-1}+o(k) & I_{n-\mu}+o(1)
\end{bmatrix},\;\; k\to0,\;\;k\in \overline{\mathbb{C}}^+.
\end{equation}
It is obvious that the right hand side of (\ref{x.1}) is invertible for $k=1$. Thus we obtain from the above assertion $(\ast)$ that the matrix $A_0':=I_\mu-B_0D_0^{-1}C_0A_0^{-1}$ is invertible.
Denote
\begin{equation*}
 T_3:=\begin{bmatrix}
I_\mu & -B_0D_0^{-1}\\  0& I_{n-\mu}
\end{bmatrix},\; T_4:=\begin{bmatrix}
A_0'^{-1} & 0\\  0& I_{n-\mu}
\end{bmatrix},\;T_5:=\begin{bmatrix}
I_\mu & 0 \\  -C_0A_0^{-1}& I_{n-\mu}
\end{bmatrix}.
\end{equation*}
Premultiplying  both sides of (\ref{x.1}) by $T_3$ and $T_4$, successively, and then postmultiplying both sides of (\ref{x.1}) by $T_5$,
 we get
 \begin{equation}\label{x.2}
  T_4T_3 M_0 J(k)T_0 T_5=
\begin{bmatrix}
kI_\mu+o(k) & o(k) \\ o(k) & I_{n-\mu}+o(1)
\end{bmatrix},\;\; k\to0,\;\;k\in \overline{\mathbb{C}}^+,
\end{equation}
which implies (\ref{3e}) with $M_1:=T_4T_3M_0$ and $T_1:=T_0 T_5$.
 Using the formula (\ref{e11}), we obtain
  \begin{equation}\label{x.3}
  T_5^{-1}T_0^{-1} J(k)^{-1}M_0^{-1}T_3^{-1} T_4^{-1}=
\begin{bmatrix}
\frac{1}{k}(I_\mu+o(1)) & o(1) \\ o(1) & I_{n-\mu}+o(1)
\end{bmatrix},\;\; k\to0,\;\;k\in \overline{\mathbb{C}}^+.
\end{equation}
 Together with (\ref{x.2}), (\ref{x.3}) and (\ref{15}), we get (\ref{3e1}).
\end{proof}

\section{Derivation of the Marchenko equation}
In this section, we shall deduce the Marchenko equation and introduce the so called normalization matrix. We note that Harmer \cite{MH3} has deduced the Marchenko equation, whereas, it is not rigorous. In this paper, with the help of (\ref{j2}) and some behaviors of the Jost solution $f(k,x)$, the Jost matrix $J(k)$ and the scattering matrix $S(k)$ presented in Propositions \ref{2.1}$-$\ref{2.3}, we will deduce the Marchenko equation rigorously.

From Proposition \ref{2.1} (c), we know that $\varphi(k,x)$ for $k\in\mathbb{R}\setminus\{0\}$ can be expressed in terms of the Jost solutions $f(\pm k,x)$. In fact, it has been proved in \cite{AW} that
\begin{equation}\label{bim}
  \varphi(k,x)=\frac{1}{2ik}f(k,x)J(-k)-\frac{1}{2ik}f(-k,x)J(k),\quad k\in\mathbb{R}\setminus\{0\}.
\end{equation}
Using (\ref{15}) and Proposition \ref{2.2} (a) and (c),  we can rewrite the above equation as
\begin{equation}\label{j1}
 -2ik \varphi(k,x)J(k)^{-1}=f(-k,x)+f(k,x)S(k),\quad k\in\mathbb{R}.
\end{equation}

From Proposition \ref{2.3}, we see that $S(k)-{U_0}\in L^2(\mathbb{R})$. Thus, by the theory of Fourier transformation, there exists  $F_S(\cdot)\in L^2(\mathbb{R})$ such that
\begin{equation}\label{18}
S(k)-{U_0}=\int_{-\infty}^\infty F_S(x)e^{-ikx}dx,\quad F_S(x)=\frac{1}{2\pi}\int_{-\infty}^\infty (S(k)-{U_0})e^{ikx}dk.
\end{equation}
Using (\ref{6}) and the first equation in (\ref{18}), and assuming $K(x,t)=0_n$ for $x>t$, we get
\begin{align}
\notag\!\!f(-k,x)+f(k,x)(S(k)-{U_0})=&e^{-ikx}I_n+\int_{-\infty}^\infty K(x,t)e^{-ikt}dt\\
\notag&+ e^{ikx}\int_{-\infty}^\infty F_S(s)e^{-iks}ds\\
&+\int_{-\infty}^\infty K(x,t)e^{ikt}dt\int_{-\infty}^\infty F_S(s)e^{-iks}ds.\label{19}
\end{align}
Multiplying both side of the above equation by $\frac{e^{iky}}{2\pi}$, integrating on $\mathbb{R}$ with respect to $k$, using the convolution theorem for the last term (since $K(x,\cdot)\in L^1(\mathbb{R})$ and $F_S(\cdot)\in L^2(\mathbb{R})$),
we obtain
\begin{align}
\notag\frac{1}{2\pi}\int_{-\infty}^\infty[f(-k,x)&+f(k,x)(S(k)-{U_0})-I_ne^{-ikx}]e^{iky}dk\\
=&K(x,y)+F_S(x+y)+\int_x^\infty K(x,t)F_S(t+y)dt.\label{20}
\end{align}
Here we have used that $K(x,\cdot)\in L^2(\mathbb{R}^+)$ which is from (\ref{7}).

Let us calculate the left hand side in the above equation. To this end, we shall first consider the following calculation.
Using the equations (\ref{an}), (\ref{j2}), (\ref{14}) and (\ref{oa}), we obtain
\begin{align}
\notag &-2ik\varphi(k,x)J(k)^{-1}-f(k,x){U_0}-e^{-ikx}I_n\\
\notag=&-2ik\cos kxAJ(k)^{\!-1}\!-2ik\!\frac{\sin kx}{k}BJ(k)^{\!-1}\!\\
\notag&-2ik\!\int_0^x\frac{\sin k(x-t)\cos kt}{k}V(t)dtAJ(k)^{\!-1}-\!f(k,x){U_0}-e^{-ikx}I_n\\
\notag&-2ik\left[\varphi(k,x)-\cos kxA\!-\!\frac{\sin kx}{k}B\! -\!\int_0^x\frac{\sin k(x-t)\cos kt}{k}V(t)dtA\right]J(k)^{-1}\\
\notag=&\cos kx[I_n+{U_0}]-i\sin kx[I_n-{U_0}]-e^{ikx}{U_0}-e^{-ikx}I_n+O\left(\frac{e^{x|{\rm Im}k|}}{k}\right)\\
=&O\left(\frac{e^{x|{\rm Im}k|}}{k}\right),\;|k|\to\infty,\;{\rm Im}k\ge0.\label{xvs}
\end{align}
Consider the contours $C_r:=\{k=re^{i\theta}:0\le\theta\le\pi\}$ and $\overline{C}_r:=C_r\cup[-r,r]$ with counter-clockwise direction. From (\ref{xvs}), we get
\begin{equation}\label{hh}
\left\|\int_{{C}_r}[-2ik\varphi(k,x)J(k)^{\!-1}\!-\!f(k,x){U_0}-e^{-ikx}I_n]{e^{iky}}dk\right\|\le c\int_{0}^\pi e^{r(x-y)\sin \theta}d\theta\to0
\end{equation}
as $r\to\infty$ if $y>x$. Since $[f(k,x){U_0}+e^{-ikx}I_n]e^{iky}$ is analytic inside the closed contour $\overline{C}_r$ and continuous on the boundary $\overline{C}_r$, we have
\begin{equation}\label{loa}
 \int_{\overline{C}_r}[f(k,x){U_0}+e^{-ikx}I_n]e^{iky}dk=0.
\end{equation}
Assume that the function $\det J(k)$ has $N$ different zeros on $i\mathbb{R}^+$, denoted by $\{ik_j\}_{j=1}^N$ with $0<k_1<k_2<...<k_N$. Using Proposition \ref{2.2} (c) with $\kappa=k_j$, and the residue theorem, together with (\ref{hh}) and (\ref{loa}), we obtain
\begin{align}
\notag\frac{1}{2\pi}&\int_{-\infty}^\infty[-2ik\varphi(k,x)J(k)^{\!-1}\!-\!f(k,x){U_0}-e^{-ikx}I_n]{e^{iky}}dk\\
=&\notag\frac{1}{2\pi}\lim_{r\to\infty}\left(\int_{\overline{C}_r}-\int_{{C}_r}\right)[-2ik\varphi(k,x)J(k)^{\!-1}\!-\!f(k,x){U_0}-e^{-ikx}I_n]{e^{iky}}dk\\
=&\notag\frac{1}{2\pi i}\lim_{r\to\infty}\int_{\overline{C}_r}2k \varphi(k,x)J(k)^{-1}e^{iky}dk\\
=&2i\sum_{j=1}^Nk_j\varphi(ik_j,x)N_{-,k_j}e^{-k_jy},\label{24}
\end{align}
where $N_{-,k_j}$ appear in (\ref{13}) with $\kappa=k_j$.

From (\ref{j1}), we see that the left hand side of the equation (\ref{20}) is equal to the left hand side of the equation (\ref{24}). Therefore, we obtain that for $y>x$
\begin{equation}\label{22}
 K(x,y)+F_S(x+y)\!+\!\int_x^\infty \!\!K(x,t)F_S\!(t+y)dt\!= \!2i\sum_{j=1}^Nk_j\varphi(ik_j,x)N_{-,k_j}e^{-k_jy}.
\end{equation}

 Let $P_j$ be the orthogonal projection onto $\ker J(ik_j)^\dagger$, $j=\overline{1,N}$.
Using (\ref{4}) and (\ref{12}), by a direct calculation, one can verify
 \begin{equation}\label{25}
\!\! \! \left\{\begin{split}
  AU^\dagger[f(ik_j,0)\!-\!if'(ik_j,0)]P_j&=[iJ(ik_j)^\dagger\!+\!f(ik_j,0)]P_j=f(ik_j,0)P_j,\\
    BU^\dagger[f(ik_j,0)\!-\!if'(ik_j,0)]P_j&=[J(ik_j)^\dagger\!+\!f'(ik_j,0)]P_j=f'(ik_j,0)P_j,
 \end{split}\right.
 \end{equation}
 which implies from (\ref{j3}) that
 \begin{equation}\label{26}
\varphi(ik_j,x)U^\dagger[f(ik_j,0)-if'(ik_j,0)]P_j=f(ik_j,x)P_j,
 \end{equation}
and from (\ref{12}) and (\ref{11}) that
\begin{equation}\label{27}
J(ik_j)U^\dagger[f(ik_j,0)-if'(ik_j,0)]P_j=0_n.
\end{equation}
Using (\ref{13}) and (\ref{j4}) with $\kappa=k_j$, and noting $J(ik_j)^{-1}J(ik_j)=I_n$,
we have
\begin{equation}\label{28}
N_{-,k_j}J(ik_j)=0_n,\quad N_{0,k_j}J(ik_j)+N_{-,k_j}\dot{J}(ik_j)=I_n.
\end{equation}
It follows from the first equation in (\ref{28}) that $J(ik_j)^\dagger N_{-,k_j}^\dagger=0_n$, and so
\begin{equation}\label{xxcv}
P_j N_{-,k_j}^\dagger=N_{-,k_j}^\dagger,
\end{equation}
or equivalently,
\begin{equation}\label{29}
N_{-,k_j}P_j=N_{-,k_j}.
\end{equation}
Using (\ref{27}) and the second equation in (\ref{28}), together with (\ref{j5}), (\ref{j6}) and (\ref{25}), we get
\begin{equation*}
\begin{split}
  &U^\dagger[f(ik_j,0)-if'(ik_j,0)]P_j\\
  =&N_{-,k_j}\dot{J}(ik_j)U^\dagger[f(ik_j,0)-if'(ik_j,0)]P_j\\
  =&N_{-,k_j}[\dot{f}'(ik_j,0)^\dagger A-\dot{f}(ik_j,0)^\dagger B]U^\dagger[f(ik_j,0)-if'(ik_j,0)]P_j\\
  =&N_{-,k_j}[\dot{f}'(ik_j,0)^\dagger f(ik_j,0)-\dot{f}(ik_j,0)^\dagger f'(ik_j,0)]P_j\\
  =&-2ik_jN_{-,k_j}A_{k_j}P_j.
\end{split}
\end{equation*}
Going back to (\ref{26}) and using (\ref{29}), we obtain
\begin{equation}\label{30}
  f(ik_j,x)P_j=-2ik_j\varphi(ik_j,x)N_{-,k_j}A_{k_j}P_j=-2ik_j\varphi(ik_j,x)N_{-,k_j}B_j,
\end{equation}
where $B_j= P_jA_{k_j}P_j+I_n-P_j$, and we have used
$$N_{-,k_j}(I_n-P_j)=N_{-,k_j}-N_{-,k_j}P_j=0_n.$$
It is obvious that $B_j$ is positive definite, and satisfies $B_jP_j=P_jB_j=P_jB_jP_j$. Hence $P_jB_j^{-1}=B_j^{-1}P_j=P_jB_j^{-1}P_j$. This implies $P_jB_j^{-\frac{1}{2}}=B_j^{-\frac{1}{2}}P_j=P_jB_j^{-\frac{1}{2}}P_j$.
Denote
\begin{equation}\label{xxcz}
  C_j:=P_j B_j^{-\frac{1}{2}},\quad j=\overline{1,N},
\end{equation}
which are called \emph{normalization matrices}.
Then it follows from (\ref{29})$-$(\ref{xxcz}) that
\begin{equation}\label{31}
2ik_j\varphi(ik_j,x)N_{-,k_j}\!=\!-f(ik_j,x)C_j^2\!=\!\left[-e^{-k_jx}-\int_x^\infty K(x,t)e^{-k_jt}dt\right]C_j^2.
\end{equation}
Together with (\ref{22}) and (\ref{31}), we obtain the following theorem.
\begin{theorem}
Assume that $V\in L_1^1(\mathbb{R}^+)$ and satisfies (\ref{2}), and let $\{-k_j^2\}_{j=1}^N$ be the eigenvalues of the problem $L(V,U)$ with $0<k_l<k_j$  if $l<j$. Then there holds the Marchenko equation
\begin{equation}\label{32}
 F(x+y)+K(x,y)+\int_x^\infty K(x,t)F(t+y)dt=0_n,\quad y>x\ge0,
\end{equation}
where
\begin{equation}\label{33}
F(x):=\sum_{j=1}^N C_j^2 e^{-k_jx}+\frac{1}{2\pi}\int_{-\infty}^{\infty}[S(k)-{U_0}]e^{ikx}dk,
\end{equation}
and the normalization matrices
\begin{equation}\label{34}
C_j:=P_j \left(P_j\int_0^\infty f(ik_j,x)^\dagger f(ik_j,x)dx P_j+I_n-P_j\right)^{-\frac{1}{2}}.
\end{equation}
\end{theorem}
\begin{definition}
The data
$\mathcal{S}:=\{S(k),k_j,C_j\}_{k\in \mathbb{R};j=\overline{1,N}}$ is called scattering data.
\end{definition}
\begin{remark}\label{r3.1}
From the above arguments, we see that
\begin{equation}\label{ali2}
C_j\ge0,\;\;  {\rm rank} C_j=\dim(\ker J(ik_j)),\;\; J(ik_j)^\dagger C_j=0_n,\quad j=\overline{1,N}.
\end{equation}
\end{remark}

\section{Properties of $F(x)$}
This section deals with the properties of the matrix-valued function $F(x)$ defined in (\ref{33}). In particular, we will show $F(\cdot)\in L^1(\mathbb{R}^+)$ and $F'(\cdot)\in L_1^1(\mathbb{R}^+)$. We shall use the Marchenko equation (\ref{32}) to obtain the integrability of $F(x)$.

Letting $x+y=z$ and $y+t=s$ in (\ref{32}), we get for $z>2x\ge0$
\begin{equation}\label{35}
F(z)=-K(x,z-x)-\int_z^\infty K(x,s+x-z)F(s)ds.
\end{equation}
From (\ref{35}) and (\ref{7}) it follows that if  $F\in L^1(\mathbb{R}^+)$ then $F(z)$ vanishes at infinite and is absolutely continuous for $z\in\mathbb{R}^+$ by the absolute continuity of integral, moreover, (\ref{32}) holds also for $y=x$.

With the help of (\ref{7}), and noting that $\sigma(\cdot)$ is a monotone decreasing function, we have
\begin{equation}\label{f1}
\|F(z)\|\le c \sigma\left(\frac{z}{2}\right)+c\int_z^\infty \sigma\left(\frac{s-z}{2}\right)\|F(s)\|ds,
\end{equation}
which implies from Gr\"{o}nwall's inequality that
\begin{equation}\label{xx}
\|F(z)\|\le c  \sigma\left(\frac{z}{2}\right)e^{c\int_z^\infty\sigma\left(\frac{s-z}{2}\right)ds}\le c\sigma\left(\frac{z}{2}\right)\in L^1(\mathbb{R}^+).
\end{equation}

Letting $y=x$ in (\ref{32}), and taking the derivative on both sides with respect to $x$, we have
\begin{align}
\notag 2F'(2x)+\frac{dK(x,x)}{dx}-K(x,x)F(2x)&+\int_x^\infty K_x(x,s)F(s+x)ds\\
 &+\int_x^\infty K(x,s)F'(s+x)ds=0.\label{36}
\end{align}
Integration by parts and using (\ref{8}), we obtain
\begin{equation}\label{37}
F'(2x)=\frac{1}{4}V(x)+K(x,x)F(2x)-\frac{1}{2}\int_x^\infty [K_x(x,s)-K_s(x,s)]F(s+x)ds.
\end{equation}
Note that $V\in L_1^1(\mathbb{R}^+)$. From (\ref{7}) and (\ref{8}) it follows that
\begin{equation}\label{zx.2}
  \|K(x,x)F(2x)\|= \frac{1}{2} \left\|\int_x^\infty V(t)F(2x)dt \right\|\le c \|F(2x)\| \in L^1(\mathbb{R}^+),
\end{equation}
and
\begin{equation}\label{zx.3}
  \!\|xK(x,x)F(2x)\|\!\le c\!\!\int_x^\infty \!\!\!x\|V(t)\|\| F(2x)\|dt\le c\| F(2x)\| \!\!\int_0^\infty\!\!\! t\|V(t)\|dt \in L^1(\mathbb{R}^+).
\end{equation}
This implies that the first two terms in right hand side in (\ref{37}) belong to $L_1^1(\mathbb{R}^+)$. We next prove the last term also belongs to $L_1^1(\mathbb{R}^+)$. By (\ref{9}), and using the fact that $\sup_{x\ge0}x^v\sigma(x)<\infty$ ($v=0,1$), we obtain
\begin{align}
\notag\left\|x^v\!\!\int_x^\infty [K_x(x,s)-K_s(x,s)]F(s+x)ds\right\|&\le cx^v\!\!\int_x^\infty \!\!\sigma (x)\sigma\left(\!\frac{x+s}{2}\!\right)\|F(s+x)\|ds\\
\notag&\le c\int_x^\infty\sigma (y) \|F(2y)\|dy\\
&\le c\sigma(x)\in L^1(\mathbb{R}^+),\;\;v=0,1.\label{zx.1}
\end{align}

On the other hand, from (\ref{16}) and (\ref{17}) it follows that
\begin{equation*}
  S(-k)-{U}_0=[S(k)-{U_0}]^\dagger,
\end{equation*}
which implies $F_S(x)^\dagger=F_S(x)$. Indeed, by the definition of $F_S$,

\begin{align}
 \notag F_S(x)^\dagger&=\frac{1}{2\pi}\left[\int_{-\infty}^\infty (S(k)-{U_0})e^{ikx}dk\right]^\dag\\
 \notag &=\frac{1}{2\pi}\int_{-\infty}^\infty [S(k)-{U_0}]^\dag e^{-ikx}dk\\
 \notag &=\frac{1}{2\pi}\int_{-\infty}^\infty [S(-k)-{U_0}]e^{-ikx}dk\\
 &=\frac{1}{2\pi}\int_{-\infty}^\infty [S(k)-{U_0}]e^{ikx}dk=F_S(x).\label{oma}
\end{align}
Since the normalization matrices $C_j$ ($j=\overline{1,N}$) are Hermite and $k_j>0$, we have $F(x)^\dag=F(x)$.

From the above arguments, we arrive at the following theorem.
\begin{theorem}\label{t4.1}
Assume that $V\in L_1^1(\mathbb{R}^+)$ and satisfies (\ref{2}), then the matrix-valued function $F(x)$ defined in (\ref{33}) is absolutely continuous on $[0,+\infty)$ and goes to zero at infinity, and satisfies $F(x)^\dag=F(x)$ and
\begin{equation}\label{38}
 \int_0^\infty\|(1+x)F'(x)\|dx<\infty.
\end{equation}
\end{theorem}

In the end of this section, let us summarize what we have know for the direct problem of the problem $L(V,U)$.
From Sections 2$-$4 it follows that if $V\in L_1^1(\mathbb{R}^+)$ and satisfies (\ref{2}) and $U$ in (\ref{4}) is unitary then the corresponding scattering data $\mathcal{S}$ satisfies the following assertions .

(I) The scattering matrix $S(k)$ is continuous and invertible on $\mathbb{R}$, and satisfies (\ref{16}) and (\ref{17}) for some unitary Hermitian matrix $U_0$.

(II) The matrix-valued function $F(x)$, defined by (\ref{33}), is differentiable on $\mathbb{R}^+$, and satisfies (\ref{38}).

(III)  For each $j=\overline{1,N}$, we have $k_j>0$ and the normalization matrix $C_j\ge0$.

\section{Self-adjointness of $V(x)$ and unitarity of $U$}

In the next several sections, we consider the inverse scattering problem which consists in recovering the potential $V(x)$ and matrix $U$ in the boundary condition from the scattering data $\mathcal{S}$. We shall seek the necessary and sufficient conditions for the scattering data to correspond to a self-adjoint potential $V(x)$ and a unitary matrix $U$.
Firstly, in this section, let us study what characteristics the scattering data should have in order to guarantee that the recovered potential $V(x)$ is self-adjoint and the recovered matrix $U$ is unitary.

Denote
\begin{equation}\label{psi}
G(k,x):=f(-k,x)+f(k,x)S(k),\quad k\in \mathbb{R}.
\end{equation}
With the help of (\ref{j3}) and (\ref{j1}), together with (\ref{3}), we have
\begin{equation}\label{vuy}
G(k,0)J(k)=-ik(U+I_n),\quad G'(k,0)J(k)=k(U-I_n),
\end{equation}
which implies
\begin{equation}\label{vuy1}
[G(k,0)+iG'(k,0)]J(k)=-2ikI_n,\quad [G(k,0)-iG'(k,0)]J(k)=-2ikU.
\end{equation}
It follows that
\begin{equation}\label{hyu}
  U=[G(k,0)-iG'(k,0)][G(k,0)+iG'(k,0)]^{-1},\quad k\in\mathbb{R}\setminus\{0\}.
\end{equation}
We note that the expression (\ref{hyu}) can be rewritten as
\begin{equation*}
  U=I_n-2iG'(k,0)[G(k,0)+iG'(k,0)]^{-1}=-I_n+2G(k,0)[G(k,0)+iG'(k,0)]^{-1},
\end{equation*}
which implies from (\ref{4}) that
\begin{equation*}
  A=G(k,0)[G(k,0)+iG'(k,0)]^{-1},\;\;B=G'(k,0)[G(k,0)+iG'(k,0)]^{-1},\;k\in\mathbb{R}\setminus\{0\}.
\end{equation*}
Moreover,
if $\det (U+I_n)\ne0$, i.e., $\det G(k,0)\ne0$ from (\ref{vuy}), then (\ref{hyu}) implies
\begin{equation*}
  U=2[I_n+iG'(k,0)G(k,0)^{-1}]^{-1}-I_n,\quad k\in\mathbb{R}\setminus\{0\}.
\end{equation*}
Similarly, if $\det (U-I_n)\ne0$, i.e., $\det G'(k,0)\ne0$, then (\ref{hyu}) implies
\begin{equation*}
  U=I_n-2[I_n-iG(k,0)G'(k,0)^{-1}]^{-1},\quad k\in\mathbb{R}\setminus\{0\}.
\end{equation*}

Theoretically, from the given scattering data $\mathcal{S}$, the potential $V$ and the unitary matrix $U$ in the boundary condition can be reconstructed by the following algorithm.
\begin{algorithm} Let the scattering data $\mathcal{S}$ be given.\\
\noindent \textbf{Step 1.} Construct the matrix-valued function $F(x)$ defined in (\ref{33}) by the given scattering data $\mathcal{S}$.\\
\noindent \textbf{Step 2.} Construct the Marchenko equation (\ref{32}) and then solve it to obtain $K(x,y)$.\\
\noindent \textbf{Step 3.} Recover the potential $V$ and the Jost solution $f(k,x)$ by (\ref{8}) and (\ref{6}), respectively. Then the unitary matrix $U$ is obtained from (\ref{hyu}).
\end{algorithm}

In this algorithm, one needs the solvability of the Marchenko equation (\ref{32}), which will be considered in the next section. Now, let us give the conditions on the scattering data in order to guarantee the self-adjointness of $V(x)$ and the unitarity of $U$.

\begin{theorem}\label{t5.2}
The conditions (I) and (III) presented in the end of Section 4 can guarantee that the potential $V(x)$ and the matrix $U$ reconstructed  in the above algorithm satisfy $V(x)=V(x)^\dag$ and $U^\dag U=I_n$.
\end{theorem}
\begin{proof}
It is known \cite[p.122]{AM} that the reconstructed potential $V(x)$ is self-adjoint if the matrix-valued function $F(x)$ is self-adjoint, which is determined by the scattering data $\mathcal{S}$. By a calculation similar to (\ref{oma}), we obtain $F(x)=F(x)^\dag$ if $S(k)^\dag=S(-k)$, $U_0^\dag=U_0$, $k_j>0$ and $C_j^\dag=C_j$. This proves the self-adjointness of $V(x)$.

Now let us prove that the matrix $U$ in (\ref{hyu}) is a unitary matrix. Since $V(x)=V(x)^\dag$, then
the Wronskian relations (\ref{10}) and (\ref{11}) hold. From these two equations and (\ref{psi}), together with $S(k)^\dag=S(k)^{-1}$, we have
\begin{equation}\label{sld}
  G(k,0)^\dag G'(k,0)-G'(k,0)^\dag G(k,0)=[G(k,x)^\dag;G(k,x)]=0_n.
\end{equation}
Using (\ref{yyy}) and (\ref{psi}), for each fixed $k\in \mathbb{R}\setminus\{0\}$, there holds,
\begin{equation}\label{als}
G(k,x)=e^{-ikx}I_n+e^{ikx}S(k)+o(1),\quad x\to\infty.
\end{equation}
It follows that the matrix
\begin{equation}\label{jhb}
  E:=G(k,0)^\dag G(k,0)+G'(k,0)^\dag G'(k,0)
\end{equation}
is positive definite. Indeed, it is obvious $E^\dag=E$. On the other hand, we shall show $\det E\ne0$. Conversely, if $\det E=0$ then there exists $0\ne\alpha\in \ker E$ such that
\begin{equation*}
\begin{split}
 0=(E\alpha,\alpha)&=(G(k,0)^\dag G(k,0)\alpha,\alpha)+(G'(k,0)^\dag G'(k,0)\alpha,\alpha)\\
 &=( G(k,0)\alpha,G(k,0)\alpha)+( G'(k,0)\alpha,G'(k,0)\alpha).\\
\end{split}
\end{equation*}
It follows that $G(k,0)\alpha=0=G'(k,0)\alpha$, which implies $G(k,x)\alpha=0$ for all $x\ge0$. This is a contradiction with (\ref{als}).

From (\ref{jhb}) it follows that
\begin{equation}\label{nma}
 E^{-\frac{1}{2}}[G(k,0)^\dag G(k,0)+G'(k,0)^\dag G'(k,0)] E^{-\frac{1}{2}}=I_n.
\end{equation}
Denote
\begin{equation}
C:=\begin{bmatrix}
G'(k,0)E^{-\frac{1}{2}} & G(k,0)E^{-\frac{1}{2}} \\ G(k,0)E^{-\frac{1}{2}} & -G'(k,0)E^{-\frac{1}{2}}
\end{bmatrix}.
\end{equation}
Then using (\ref{sld}) and (\ref{nma}), by a direct calculation, we get $C^\dag C=I_{2n}$. This yields $CC^\dag =I_{2n}$, which implies
\begin{equation*}
 \left\{\begin{split}
   &G'(k,0)E^{-1}G'(k,0)^\dag+ G(k,0)E^{-1}G(k,0)^\dag =I_n, \\
     & G'(k,0)E^{-1}G(k,0)^\dag-G(k,0)E^{-1}G'(k,0)^\dag =0_n.
 \end{split}\right.
\end{equation*}
Therefore,
\begin{align}
\notag  [G(k,0)&+iG'(k,0)]E^{-1}[G(k,0)^\dag-iG'(k,0)^\dag]\\
\notag   =&G(k,0)E^{-1}G(k,0)^\dag+G'(k,0)E^{-1}G'(k,0)^\dag\\
  &-i[G(k,0)E^{-1}G'(k,0)^\dag-G'(k,0)E^{-1}G(k,0)^\dag]
  =I_n\label{loq}
\end{align}

Consider the adjoint matrix of $U$ defined by (\ref{hyu}),
\begin{equation*}
  U^\dag=[G(k,0)^\dag-iG'(k,0)^\dag]^{-1}[G(k,0)^\dag+iG'(k,0)^\dag],\quad k\in\mathbb{R}\setminus\{0\}.
\end{equation*}
It follows from (\ref{loq}) that
\begin{equation*}
  \begin{split}
  U^\dag U=&[G(k,0)^\dag-iG'(k,0)^\dag]^{-1}[G(k,0)^\dag+iG'(k,0)^\dag]\\
  &\times[G(k,0)-iG'(k,0)][G(k,0)+iG'(k,0)]^{-1}\\
  =&[G(k,0)^\dag-iG'(k,0)^\dag]^{-1}E[G(k,0)+iG'(k,0)]^{-1}=I_n.
  \end{split}
\end{equation*}
Consequently, the matrix of $U$ defined by (\ref{hyu}) is unitary.
\end{proof}

\begin{remark}
Let $k_j>0$ and $C_j^\dag=C_j$ for $j=\overline{1,N}$. From the the proof of Theorem \ref{t5.2}, we see that, the only conditions that $S(k)^\dag=S(-k)$ and $U_0^\dag=U_0$ can guarantee the self-adjointness of the potential $V(x)$, and the only condition $S(k)^\dag=S(k)^{-1}$ can guarantee the unitarity of the matrix $U$.
\end{remark}

\section{Solvability of the Marchenko equation}
In this section, we study the solvability of the Marchenko equation (\ref{32}). We will give the conditions on the data $\{S(k),k_j,C_j\}_{k\in \mathbb{R};j=\overline{1,N}}$ such that the Marchenko equation (\ref{32}) is uniquely solvable.

It is known \cite{AM,BL} that if the matrix-valued function $F(x)$ satisfies the condition (II) presented in Section 4, then $F(\cdot)\in L^1(\mathbb{R}^+)$ and the operator $\mathbf{F}_x$, defined by
\begin{equation*}
  (\mathbf{{F}}_xh)(t):=\int_x^\infty h(s)F(s+t)ds,
\end{equation*}
is compact in $L^1((x,+\infty),\mathbb{C}^n)$ for each $x\ge0$. Thus, by the Fredholm alternative theorem, we shall seek the conditions that guarantee the following homogeneous equation (\ref{41}) has only zero solution in $L^1((x,+\infty),\mathbb{C}^n)$ for each fixed $x\ge0$.
\begin{lemma}\label{l5.1}
Assume that the matrix-valued function $F(t)$ satisfies the condition (II) presented in Section 4.
For each fixed $x\ge0$, if $h\in L^1((x,\infty);\mathbb{C}^n)$ solves
\begin{equation}\label{41}
  h(t)+\int_x^\infty h(s)F(s+t)ds=0,
\end{equation}
then $h(t)$ satisfies the following properties: (a) $h\in L^2(x,+\infty)$; (b) $h(t)$ is continuous on $[x,+\infty)$; (c) $h(t)\to0$ as $t\to+\infty$; (d) the derivative $h'(t)$ exists and belongs to $L^1(x,\infty)$.
\end{lemma}

\begin{proof}
(a) Since $F(t)$ satisfies (\ref{38}), we have
\begin{equation}\label{gun}
 \| F(t)\|\le c\int_t^\infty \|F'(s)\|ds\to0, \;\;t\to+\infty.
\end{equation}
Note that $F(t)$ is continuous on $[0,+\infty)$.
It follows from (\ref{gun}) that  $\sup_{t\ge 0}F(t)<\infty$.
Postmultiplying both sides of (\ref{41}) by $h(t)^\dag$, and integrating on $[x,\infty)$, we have
\begin{equation*}
\begin{split}
 \int_x^\infty h(t)h(t)^\dag dt
& \le c \int_x^\infty\int_x^\infty \|h(s)\|\|F(s+t)\|\|h(t)\|dsdt\\
&\le c\sup_{t\ge x,s\ge x}\|F(s+t)\|\left[\int_x^\infty \|h(s)\|ds\right]^2<\infty,
\end{split}
\end{equation*}
which implies that if $h\in L^1(x,+\infty)$ solves (\ref{41}) then $h\in L^2(x,+\infty)$.

(b) Since $F(t)$ is uniformly continuous on $[x,\infty)$, we have
\begin{equation*}
 \begin{split}
 \| h(t+\Delta t)-h(t)\|&\le c\int_x^\infty \|h(s)\|\|F(t+\Delta t+s)-F(t+s)\|ds\\
 &\le c\sup_{s\ge x}\|F(t+s+\Delta t)-F(t+s)\|\to0,\;\Delta t\to0,
 \end{split}
\end{equation*}
which implies that $h(t)$ is continuous on $[x,+\infty)$.

(c) By virtue of (\ref{gun}), we have
\begin{equation*}
   \| h(t)\|\le c \sup_{s\ge x}\|F(t+s)\|\int_x^\infty \|h(s)\|ds\to0,\;t\to+\infty.
\end{equation*}

(d) Since $F'(t)$ exists and belongs to $L^1(x,+\infty)$, there exists $\theta\in[0,1]$ such that
\begin{align}
\notag \frac{h(t+\Delta t)-h(t)}{\Delta t}=&-\int_x^\infty h(s)\frac{F(t+\Delta t+s)-F(t+s)}{\Delta t}ds.\\
\notag =&-\int_x^\infty h(s)F'(t+\theta\Delta t+s)ds\\
\notag =&-\int_{x+t+\theta\Delta t}^\infty h(s-t)F'(s)ds\\
  &-\int_{x+t+\theta\Delta t}^\infty [h(s-t-\theta\Delta t)-h(s-t)]F'(s)ds\label{jkx}
\end{align}
By (b) and (c), we see that $h(t)$ is uniformly continuous on $[x,+\infty)$. This implies the last term in (\ref{jkx}) vanishes as $\Delta t\to0$. It follows that
\begin{equation*}
\begin{split}
  h'(t)=-\int_{x+t}^\infty h(s-t)F'(s)ds=-\int_{x}^\infty h(s)F'(s+t)ds.
\end{split}
\end{equation*}
Furthermore, we have
\begin{equation*}
\begin{split}
 \int_x^\infty \|h'(t)\|&\le c\int_x^\infty\int_x^\infty\|h(s)\|\|F'(s+t)\|dsdt\\
 &\le c\int_x^\infty\int_{x+s}^\infty\|h(s)\|\|F'(t)\|dtds<\infty.
\end{split}
\end{equation*}
This completes the proof.
\end{proof}

\begin{lemma}\label{l5.2}
Assume that the data $\{S(k),k_j,C_j\}_{k\in\mathbb{R};j=\overline{1,N}}$ satisfies the conditions (I)$-$(III) presented in Section 4. For each fixed $x\ge0$, if the row vector-valued function $h(t)\in L^1((x,+\infty),\mathbb{C}^n)$ solves (\ref{41}), then
\begin{equation}\label{46}
\widehat{h}(-ik_j)C_j=0,\quad j=\overline{1,N},
\end{equation}
and
\begin{equation}\label{bui}
\widehat{ h}(-k)S(k)=-\widehat{h}(k),\quad k\in \mathbb{R},
\end{equation}
where
\begin{equation}\label{42}
\widehat{ h}(k):=\int_x^\infty h(t)e^{-ikt}dt.
\end{equation}
\end{lemma}
\begin{proof}
 Let $h\in L^1((x,+\infty),\mathbb{C}^n)$ with $h(t)=0$ for $t<x$ solve the equation (\ref{41}). Then from Lemma \ref{l5.1} it follows that $h\in L^2(\mathbb{R})$ with $h(t)=0$ for $t<x$. Denote
\begin{equation*}
 g(t):=\int_{x}^\infty h(s)F_S(s+t)ds=\int_{-\infty}^\infty h(-s)F_S(t-s)ds.
\end{equation*}
Since $F_S\in L^2(\mathbb{R})$ and (\ref{18}), we have
\begin{equation}\label{40}
\widehat{ g}(k)=\widehat{ h}(-k)[S(k)-U_0].
\end{equation}
Postmultiplying both sides of (\ref{41}) by $h(t)^\dag$, integrating on $\mathbb{R}$, and using (\ref{33}), we get
 \begin{equation}\label{43}
   \int_{-\infty}^\infty h(t)h(t)^\dag dt+\int_{-\infty}^\infty g(t)h(t)^\dag dt+\sum_{j=1}^N  \widehat{h}(-ik_j)C_j^2\widehat{h}(-ik_j)^\dag=0.
 \end{equation}
 Note that
 \begin{equation*}
  \int_{-\infty}^\infty \!h(t)h(t)^\dag dt\!=\!\frac{1}{2\pi} \!\!\int_{-\infty}^\infty \widehat{h}(k)\widehat{h}(k)^\dag dk,\;\int_{-\infty}^\infty \!g(t)h(t)^\dag dt\!=\!\frac{1}{2\pi}\!\! \int_{-\infty}^\infty \widehat{g}(k)\widehat{h}(k)^\dag dk,
 \end{equation*}
and $C_j^\dag=C_j$. It follows from (\ref{40}) that the equation (\ref{43}) is equivalent to
\begin{align}
\notag\int_{-\infty}^\infty \!\widehat{h}(k)\widehat{h}(k)^\dag dk\!+\!\int_{-\infty}^\infty \!\widehat{ h}(-k)&[S(k)\!-U_0]\widehat{h}(k)^\dag dk\\
&+2\pi\sum_{j=1}^N [\widehat{h}(-ik_j)C_j][\widehat{h}(-ik_j)C_j]^\dag\!=0.\label{44}
\end{align}
Since $U_0$ is a Hermitian matrix and $h(t)=0$ for $t<x$, we have
\begin{equation*}
\int_{-\infty}^\infty \!\widehat{ h}(-k)U_0\widehat{h}(k)^\dag dk=2\pi\int_{-\infty}^\infty \!{ h}(-t)U_0{h}(t)^\dag dt=0.
\end{equation*}
Therefore, we can rewrite (\ref{44}) as
\begin{equation}\label{45}
 \int_{-\infty}^\infty \![\widehat{h}(k)+\widehat{ h}(-k)S(k)]\widehat{h}(k)^\dag dk+2\pi\sum_{j=1}^N [\widehat{h}(-ik_j)C_j][\widehat{h}(-ik_j)C_j]^\dag\!=0.
\end{equation}
Note that the second term in the above equation is non-negative.
By the Schwarz inequality and the unitarity of $S(k)$, we get
\begin{equation*}
\begin{split}
\left|\left(\widehat{h}(-k)S(k),\widehat{h}(k)\right)\right|&\le \left(\widehat{h}(-k),\widehat{h}(-k)\right)^{\frac{1}{2}}\left(\widehat{h}(k),\widehat{h}(k)\right)^{\frac{1}{2}}\\
&\le\frac{1}{2}\left[\left(\widehat{h}(-k),\widehat{h}(-k)\right)+\left(\widehat{h}(k),\widehat{h}(k)\right)\right],
\end{split}
\end{equation*}
which implies
\begin{equation*}
  \left| \int_{-\infty}^\infty \widehat{ h}(-k)S(k)\widehat{h}(k)^\dag dk\right|\le \int_{-\infty}^\infty \left|\widehat{ h}(-k)S(k)\widehat{h}(k)^\dag\right| dk\le \int_{-\infty}^\infty \widehat{h}(k)\widehat{h}(k)^\dag dk.
\end{equation*}
It follows that the real part of the first term in (\ref{45}) is also non-negative. Thus, we conclude that
(\ref{46}) holds,
and hence
\begin{equation}\label{47}
\int_{-\infty}^\infty \![\widehat{h}(k)+\widehat{ h}(-k)S(k)]\widehat{h}(k)^\dag dk=0.
\end{equation}
Denote $z(k):=\widehat{h}(k)+\widehat{ h}(-k)S(k)$. Then using (\ref{47}), we have
\begin{equation*}
\begin{split}
\int_{-\infty}^\infty \!\widehat{h}(k)\widehat{h}(k)^\dag dk&=\int_{-\infty}^\infty \!\widehat{h}(-k)\widehat{h}(-k)^\dag dk\\
&=\int_{-\infty}^\infty \![\widehat{h}(-k)S(k)][\widehat{h}(-k)S(k)]^\dag dk\\
&=\int_{-\infty}^\infty \![z(k)-\widehat{h}(k)][z(k)-\widehat{h}(k)]^\dag dk\\
&=\int_{-\infty}^\infty \!z(k)z(k)^\dag dk+\int_{-\infty}^\infty \!\widehat{h}(k)\widehat{h}(k)^\dag dk,
\end{split}
\end{equation*}
which is impossible unless $\int_{-\infty}^\infty \!z(k)z(k)^\dag dk=0$, which implies (\ref{bui}).
\end{proof}

Since $U_0$ is a unitary Hermitian matrix, without loss of generality, in the following discussion, we assume
\begin{equation}\label{xmk}
U_0=M{\rm diag}\{I_{n-n_D},-I_{n_D} \}M^\dag
\end{equation}
for some unitary matrix $M$. Let us show that this assumption is reasonable. Since $U_0=U_0^\dag=U_0^{-1}$, the eigenvalues of which are either $1$ or $-1$. Thus there exists a unitary matrix $M'$ such that $M'^\dag U_0 M'={\rm diag}\{a_1,...,a_i,...,a_j,...,a_n \}:=D$, where $a_j=\pm1$, $j=\overline{1,n}$. Note that premultiplying the matrix $D$ by a permutation matrix $P(i,j)$, one can exchange the $i$-th row and $j$-th row of $D$, and postmultiplying  $D$ by  $P(i,j)$, one can exchange the $i$-th column and $j$-th column of $D$, and hence $P(i,j)DP(i,j)={\rm diag}\{a_1,...,a_j,...,a_i,...,a_n \}$. On the other hand, note that every permutation matrix is also unitary Hermitian matrix, and the product of two unitary matrices is also unitary. Thus the assumption (\ref{xmk}) is reasonable.

\begin{theorem}\label{t5.1}
Assume that the data $\{S(k),k_j,C_j\}_{k\in\mathbb{R};j=\overline{1,N}}$ satisfies the conditions (I)$-$(III) presented in the end of Section 4, and the condition (IV): the matrix $S(k)=-J_1(-k)J_1(k)^{-1}$ for some matrix-valued function $J_1(k)$, which is invertible for $k\in\mathbb{R}\setminus\{0\}$ and satisfies:\\
(i) $J_1(k)$ is analytic in $\mathbb{C}^+$, continuous in $\overline{\mathbb{C}}^+$ and satisfies
\begin{equation}\label{xmk1}
  J_1(k)=-ik A_1+B_1+QA_1+o(1),\quad |k|\to\infty,\;\;k\in \overline{\mathbb{C}}^{+},
\end{equation}
where $Q$ is a constant matrix and
\begin{equation}\label{xmk2}
A_1=M{\rm diag}\{I_{n-n_D},0_{n_D}\}M^\dag,\;B_1=M{\rm diag}\{0_{n-n_D},-iI_{n_D}\}M^\dag;
\end{equation}\\
 (ii) $C_j J_1(ik_j)=0$ and $\dim(\ker J_1(ik_j))={\rm rank}C_j$ ($j=\overline{1,N}$);\\
  (iii) $J_1(k)^{-1}$ has simple poles at $k=ik_j$ ($j=\overline{1,N}$) and maybe has a simple pole $k=0$, \\
  then the equation (\ref{41}) has only zero solution in $L^1((x,+\infty);\mathbb{C}^n)$ for each fixed $x\ge0$.
\end{theorem}
\begin{proof}
Assume that $h(t)$ solves the equation (\ref{41}) in $L^1((x,+\infty);\mathbb{C}^n)$ for $x\ge0$.
From (\ref{bui}), we have
\begin{equation}\label{goq}
\widehat{h}(-k)J_1(-k)J_1(k)^{-1}=\widehat{h}(k),\quad k\in \mathbb{R}.
\end{equation}
Since $S(k)=S(-k)^\dag$, we obtain
\begin{equation}\label{bkx}
   J_1(-k)J_1(k)^{-1}=[J_1(-{k})^\dag]^{-1}J_1({k})^\dag=[J_1(-\bar{k})^\dag]^{-1}J_1(\bar{k})^\dag,\quad k\in \mathbb{R}.
\end{equation}
Together with (\ref{goq}) and (\ref{bkx}) we get
\begin{equation}\label{xxc1}
k\widehat{h}(-k)[J_1(-\bar{k})^\dag]^{-1}=k\widehat{h}(k)[J_1(\bar{k})^\dag]^{-1},\quad k\in \mathbb{R}.
\end{equation}
Here, to add the factor $k$ is because $[J_1(-\bar{k})^\dag]^{-1}$ maybe has a simple pole $k=0$, and this manipulation can guarantee that both sides of (\ref{xxc1}) are continuous at $k=0$.

Let us prove that the left and right hand sides of (\ref{xxc1}) are analytic in the upper and lower-half complex plane, respectively. If it true, then
the function
\begin{equation}\label{dp}
  P(k):=\left\{\begin{split}
         &k\widehat{h}(-k)[J_1(-\bar{k})^\dag]^{-1},\quad  k\in \overline{\mathbb{C}}^+, \\
          & k\widehat{h}(k)[J_1(\bar{k})^\dag]^{-1},\quad  k\in \overline{\mathbb{C}}^-,
       \end{split}\right.
\end{equation}
is analytic in the whole complex plane.

In fact, by virtue of the assumptions (i) and (iii), we get that $[J_1(-\bar{k})^\dag]^{-1}$ is analytic in ${\mathbb{C}}^+$ except for a finite number of simple poles $\{ik_j\}_{j=1}^N$, in particular,
\begin{equation*}\label{oun}
 [J_1(-\bar{k})^\dag]^{-1}=\frac{N_{-,k_j}^\dag}{-k+ik_j}+N_{0,k_j}^\dag+O(-k+ik_j),\quad k\to ik_j,\;j=\overline{1,N}.
\end{equation*}
Note that $J_1(-\bar{k})^\dag[J_1(-\bar{k})^\dag]^{-1}=I_n$, which implies $J_1(ik_j)^\dag N_{-,k_j}^\dag=0$. Then, we have
\begin{equation}\label{gfl}
  P_{1j} N_{-,k_j}^\dag=N_{-,k_j}^\dag,
\end{equation}
where $P_{1j}$ is the orthogonal projection onto $\ker J_1(ik_j)^\dagger$. On the other hand, by the  assumption (ii), we have
\begin{equation}\label{gfl1}
C_j=P_{1j}C_j=P_{1j}(C_j+I_n-P_{1j}):=P_{1j}D_j,
\end{equation}
where $D_j=C_j+I_n-P_{1j}$ is invertible. Together with (\ref{gfl1}), (\ref{gfl}) and (\ref{46}), we have
\begin{equation*}\label{hjan}
\widehat{h}(-ik_j){N_{-,k_j}^\dag}=\widehat{h}(-ik_j)P_{1j}{N_{-,k_j}^\dag}=\widehat{h}(-ik_j)C_jD_j^{-1}{N_{-,k_j}^\dag}=0.
\end{equation*}

Note that $\widehat{h}(-k)$ is analytic in $\mathbb{C}^+$.
It follows that $k\widehat{h}(-k)[J_1(-\bar{k})^\dag]^{-1}$ is analytic in the upper half plane. Similarly, one can also prove $k\widehat{h}(k)[J_1(\bar{k})^\dag]^{-1}$ is analytic in the lower half plane.
Therefore, $P(k)$ is analytic in $\mathbb{C}$.

By virtue of Lemma \ref{l5.1}, we get
\begin{equation*}
 \widehat{ h}(k)=-\frac{1}{ik}\left[\left.{h(t)e^{-ikt}}\right|_x^\infty-\int_x^\infty h'(t)e^{-ikt}dt\right]=\frac{\widehat{h'}(k)+h(x)e^{-ikx}}{ik},
\end{equation*}
which implies $k\widehat{ h}(k)\to-ih(x)e^{-ikx}$ as $|k|\to\infty$ in $\overline{\mathbb{C}}^{-}$. Similarly, we can also get $k\widehat{ h}(-k)\to ih(x)e^{ikx}$ as $|k|\to\infty$ in $\overline{\mathbb{C}}^{+}$.
On the other hand, note that $[J_1(-\bar{k})^\dag]^{-1}=O(1)$ when $|k|\to\infty$ in $\overline{\mathbb{C}}^{+}$ (see (\ref{xmk1}) and (\ref{xmk2})). Thus, $P(k)$ is a bounded row vector-valued entire function, and so by Liouville's theorem, $P(k)\equiv P_0$ which is a constant row vector. It follows from (\ref{dp}) that
\begin{equation}\label{kwq}
\left\{\begin{split}
&k\widehat{ h}(-k)=P_0J_1(-\bar{k})^\dag,\quad k\in\overline{\mathbb{C}}^{+},\\
&k\widehat{ h}(k)=P_0J_1(\bar{k})^\dag,\quad k\in\overline{\mathbb{C}}^{-},
\end{split}\right.
\end{equation}
which implies from  (\ref{xmk1}) that
\begin{equation*}
\left\{\begin{split}
 & ih(x)e^{ikx}-P_0(-ikA_1^\dag+B_1^\dag+A_1^\dag Q^\dag)\to0,\quad k\to+\infty,\;k\in \mathbb{R},\\
 & -ih(x)e^{ikx}-P_0(ikA_1^\dag+B_1^\dag+A_1^\dag Q^\dag)\to0,\quad k\to+\infty,\;k\in \mathbb{R}.
  \end{split}\right.
\end{equation*}
Note that as $k\to+\infty$ the limits of $e^{\pm ikx}$ do not exist if $x>0$ and $\pm ik P_0A_1^\dag\to\infty$ if $P_0A_1^\dag\ne0$. Thus $h(x)=0$ if $x>0$ and $P_0A_1^\dag=0$. In the case $x=0$, we have $\pm ih(0)=P_0B_1^\dag$, which implies $h(0)=0$.
This completes the proof.
\end{proof}

\begin{theorem}\label{t6.2}
Under the assumptions in Theorem \ref{t5.1}, for each fixed $x\ge0$, the Marchenko (\ref{32}) has a unique solution $K(x,y)$, and the matrix-valued function
\begin{equation}\label{xx.1}
  f(k,x):=e^{ikx}+\int_x^\infty K(x,t)e^{ikt}dt,\;\;x\ge0
\end{equation}
satisfies the Schr\"{o}dinger equation $-f''+V(x)f=k^2f$ with the potential
\begin{equation}\label{xx.2}
  V(x)=-2\frac{d}{dx}K(x,x),\quad x\ge0,
\end{equation}
and $V\in L_1^1(\mathbb{R}^+)$.
\end{theorem}
\begin{proof}
As mentioned in the beginning of this section, that the equation (\ref{41}) has only zero solution in $L^1$ can guarantee the unique solvability of the Marchenko (\ref{32}) in $L^1$ by the Fredholm alternative theorem. It is shown in \cite[p.117-122]{AM} that if the Marchenko (\ref{32}) has a unique solution $K(x,y)$ in $L^1$, then it satisfies (\ref{xx.1}), (\ref{xx.2}) and
\begin{equation}\label{zx.0}
 \|V(x)-4F'(2x)\|< \infty,
\end{equation}
which implies from the condition (II) that
$V\in L_1^1(\mathbb{R}^+)$.
\end{proof}
\section{Necessary and sufficient conditions}
In this section, we shall provide the characterization of the scattering data $\mathcal{S}=\{S(k),k_j,C_j\}_{k\in \mathbb{R};j=\overline{1,N}}$. In other
words, necessary and sufficient conditions for the solvability of the inverse scattering problem $L(V,U)$ are characterized.
  In Theorem \ref{t5.1}, we have given the sufficient conditions for the unique solvability of the Marchenko equation (\ref{32}), which depend on the matrix-valued function $J_1(k)$. In this section, for the direct problem, we will first show the existence of the matrix $J_1(k)$. Then we consider the uniqueness of the inverse scattering problem. Finally, the necessary and sufficient conditions are provided.

\begin{lemma}\label{l5.3}
Assume that $V\in L_1^1(\mathbb{R}^+)$ and satisfies (\ref{2}), and  the matrix $U$ in (\ref{4}) is unitary. Then there exists the matrix-valued function $J_1(k)$ satisfying $S(k)=-J_1(-k)J_1(k)^{-1}$ and the conditions (i)$-$(iii) in Theorem \ref{t5.1}. Moreover, the rank of the matrix $J_1(0)$ is $n-\mu$, where $\mu$ is the multiplicity of the eigenvalue 1 of the matrix $S(0)$.
\end{lemma}
\begin{proof}
By virtue of (\ref{4}), (\ref{14}) and Proposition \ref{2.3} (d), together with  (\ref{xmk1}), we have that as $|k|\to\infty$ in $\overline{\mathbb{C}}^+$
\begin{align}
\notag 2 M^\dag J(k) M=&-{ik} M^\dag(U+I_n) M+iM^\dag(U-I_n) M\\
\notag&+M^\dag Q(0)M M^\dag(U+I_n) M+o(1)\\
\notag =&-ik{\rm diag}\{I_{n-n_D},0_{n_D}\}T^{-1}+i{\rm diag}\{(X_1-I_{n-n_D}),-2I_{n_D}\}\\
\quad &+Q'{\rm diag}\{I_{n-n_D},0_{n_D}\}T^{-1}+o(1),\label{ali}
\end{align}
where $T={\rm diag}\{(X_1+I_{n-n_D})^{-1},\frac{1}{2}I_{n_D}\}$ and $X_1$ appear in Proposition \ref{2.3} (d), and $Q'=M^\dag Q(0)M$. Postmultiplying both sides of (\ref{ali}) by $T$, we get that as $|k|\to\infty$ in $\overline{\mathbb{C}}^+$
\begin{align}
\notag 2M^\dag J(k) MT=&-ik{\rm diag}\{I_{n-n_D},0_{n_D}\}\\
\notag&+i{\rm diag}\{(X_1-I_{n-n_D})(X_1+I_{n-n_D})^{-1},-I_{n_D}\}\\
\notag&+Q'{\rm diag}\{I_{n-n_D},0_{n_D}\}+o(1)\\
\notag=&-ik{\rm diag}\{I_{n-n_D},0_{n_D}\}+i{\rm diag}\{0_{n-n_D},-I_{n_D}\}\\
&+Q{\rm diag}\{I_{n-n_D},0_{n_D}\}+o(1),\label{ali1}
\end{align}
where $Q=i{\rm diag}\{(X_1-I_{n-n_D})(X_1+I_{n-n_D})^{-1},0_{n_D}\}+Q'$.

Denote $T_6:=2MTM^\dag$, and $J_1(k):=J(k)T_6$. Then from (\ref{ali1}) and (\ref{ali2}) together with Proposition \ref{2.2} we obtain that $J_1(k)$ satisfies the conditions (i)$-$(iii) in Theorem \ref{t5.1} with $A_1$ and $B_1$ satisfying (\ref{xmk2}), and ${\rm rank} J_1(0)=n-\mu$ from (\ref{3e}).
\end{proof}

Now let us consider the uniqueness of the inverse scattering problem. Together with the problem $L(V,U)$, we consider the problem $L(\widetilde{V},\widetilde{U})$ of the same form but with different coefficients $\widetilde{V}(x)$ and $\widetilde{U}$. We agree that if a certain symbol  $\delta$ denotes an object related to $L(V,U)$, then $\widetilde{\delta}$ will denote an analogous object related to $L(\widetilde{V},\widetilde{U})$.

\begin{theorem}\label{t7.1}
Assume that $V(x)$ and $\widetilde{V}(x)$ satisfy (\ref{2}) in $L_1^1(\mathbb{R}_+)$, and $U$ and $\widetilde{U}$ are unitary matrices. If $\mathcal{S}=\mathcal{\widetilde{S}}$ then $V(x)=\widetilde{V}(x)$ and $U=\widetilde{U}$.
\end{theorem}
\begin{proof}
From (\ref{6}), (\ref{psi}) and (\ref{hyu}), we see that the matrix $U$ in the boundary condition (\ref{3}) can be uniquely determined by the scattering matrix $S(k)$ and the kernel $K(x,y)$.
Moreover, the potential $V(x)$ is recovered by $K(x,y)$ (see (\ref{8})) and the matrix-valued function $F(x)$ is uniquely determined by scattering data $\mathcal{S}$ (see (\ref{33})).
Therefore, it is sufficient to prove that $K(x,y)$ can be uniquely determined by $F(x)$, together with (\ref{32}), which is equivalent to show that  for each $x\ge0$ the equation (\ref{41})
has only zero solution in $L^1((x,\infty);\mathbb{C}^n)$. Under the assumption that $V$ satisfies (\ref{2}) and $U$ is unitary, the scattering data $\mathcal{S}$ satisfies the conditions (I)$-$(III) presented in Section 4. Using Theorem \ref{t5.1}, it is enough to prove that the scattering data $\mathcal{S}$ uniquely determine $J_1(k)$ satisfying the conditions (i)$-$(iii). Namely, we shall prove that if $\mathcal{S}=\mathcal{\widetilde{S}}$ then $J_1(k)=\widetilde{J_1}(k)$ for all $k\in \overline{\mathbb{C}}^+$.

From (\ref{15}) and Lemma \ref{l5.3}, together with $S(k)=\widetilde{S}(k)$, we have
\begin{equation*}
J_1(-k)J_1(k)^{-1}=\widetilde{J}_1(-k)\widetilde{J}_1(k)^{-1},\quad k\in\mathbb{R}.
\end{equation*}
Since $J_1(k)$ is invertible for $k\in \mathbb{R}\setminus\{0\}$ and it is possible $\det J_1(0)=0$, we get
\begin{equation}\label{cxy}
\widetilde{J}_1(-k)^{-1}J_1(-k)=\widetilde{J}_1(k)^{-1}J_1(k),\quad k\in\mathbb{R}\setminus\{0\}.
\end{equation}
By Lemma \ref{l5.3}, we know that $J_1(k)^{-1}$ and $\widetilde{J}_1(k)^{-1}$ are analytic in $\mathbb{C}^+$ except for the simple poles $\{ik_j\}_{j=1}^N$. Thus, we have
\begin{equation*}\label{oun}
 J_1({k})^{-1}=\frac{N_{1-,k_j}}{k-ik_j}+O(1),\; \widetilde{J}_1({k})^{-1}=\frac{\widetilde{N}_{1-,k_j}}{k-ik_j}+O(1),\quad k\to ik_j,\;j=\overline{1,N}.
\end{equation*}
Using the same method as in (\ref{gfl}) and (\ref{gfl1}), we can get
\begin{equation}\label{oxn}
N_{1-,k_j}=N_{1-,k_j}[D_j^{-1}]^\dag C_j,\quad \widetilde{N}_{1-,k_j}=\widetilde{N}_{1-,k_j}[\widetilde{D}_j^{-1}]^\dag C_j.
\end{equation}
It follows from (\ref{oxn}) and the assumption (ii) in Theorem \ref{t5.1} that
\begin{equation*}
 \widetilde{N}_{1-,k_j} J_1(ik_j)=0,\quad j=\overline{1,N}.
\end{equation*}
Note that $J(k)$ and $\widetilde{J}(k)$ are analytic in $\mathbb{C}^+$. It follows that
the left hand side of (\ref{cxy}) is analytic in $\mathbb{C}^{-}$ and right hand side of (\ref{cxy}) is analytic in $\mathbb{C}^{+}$.

Let us prove that the equation (\ref{cxy}) also holds at $k=0$. Without loss of generality, we assume that $J(k), S(0)$  and $\widetilde{J}(k), \widetilde{S}(0)$ have the forms of (\ref{3e}) and (\ref{3e1}), respectively, for some invertible matrices $M_1,T_1$ and $\widetilde{M}_1,\widetilde{T}_1$. Then it follows from the proof of Lemma \ref{l5.3} that
\begin{equation}\label{fkjz}
  J_1(k)=J(k)T_6=kM_1^{-1} \begin{bmatrix}
I_\mu+o(1) & o(1) \\ o(1) & \frac{1}{k}(I_{n-\mu}+o(1))
\end{bmatrix}T_1^{-1}T_6, \quad k\to0,\;\;k\in \overline{\mathbb{C}}^+.
\end{equation}
Similarly, we can get the asymptotics of $\widetilde{J_1}(k)$ as $k\to0$.
Using the formula (\ref{e11}), we calculate the inverse of $\widetilde{J_1}(k)$, and get
\begin{equation}\label{fkjz1}
\!\!  \widetilde{J}_1(k)^{-1}\!=\!\frac{1}{k}\widetilde{T}_6^{-1}\widetilde{T}_1 \begin{bmatrix}
I_\mu+o(1) & o(k) \\ o(k) & kI_{n-\mu}+o(k)
\end{bmatrix}\widetilde{M}_1, \;\; k\to0,\;\;k\in \overline{\mathbb{C}}^+.
\end{equation}
Denote $
 \widetilde{M}_1M_1^{-1}:=\begin{bmatrix}
W_1 & W_2 \\ W_3 & W_4
\end{bmatrix}.$
Since $S(0)=\widetilde{S}(0)$ and (\ref{3e1}), we have
\begin{equation*}
\begin{bmatrix}
W_1 & W_2 \\ W_3 & W_4
\end{bmatrix}\begin{bmatrix}
I_\mu & 0 \\ 0& -I_{n-\mu}
\end{bmatrix}=\begin{bmatrix}
I_\mu & 0 \\ 0& -I_{n-\mu}
\end{bmatrix}\begin{bmatrix}
W_1 & W_2 \\ W_3 & W_4
\end{bmatrix}.
\end{equation*}
By a direct calculation, we get $W_2=0$ and $W_3=0$. It follows from
(\ref{fkjz}) and  (\ref{fkjz1})  that $\widetilde{J}(k)^{-1}J(k)$ is well defined at $k=0$. Hence the equation (\ref{cxy}) also holds at $k=0$.
Denote
\begin{equation}\label{xcoa}
  Y(k):=\left\{\begin{split}
         &\widetilde{J}_1(-k)^{-1}J_1(-k), \quad k\in \overline{\mathbb{C}}^-,\\
          & \widetilde{J}_1(k)^{-1}J_1(k),\quad k\in \overline{\mathbb{C}}^+.
        \end{split}\right.
\end{equation}
Then by  analytic continuation theorem the matrix-valued function $Y(k)$ is analytic in the whole complex plane.

Denote $J_1^0(k):=-ikA_1+B_1$, which is obviously invertible for large $k$ from (\ref{xmk2}). Then rewrite (\ref{xmk1}) as
\begin{equation}\label{xmk3}
  J_1(k)=J_1^0(k)-Q\frac{J_1^0(k)}{ik}+\frac{QB_1}{ik}+o(1),\quad |k|\to\infty,\;\;k\in \overline{\mathbb{C}}^{+}.
\end{equation}
Note that for $|k|\to\infty$ in $\overline{\mathbb{C}}^{+}$ there holds
$J_1^0(k)^{-1}=O(1)$.
Thus, we have $J_1(k)J_1^0(k)^{-1}=I_n+o(1)$, or equivalently,
\begin{equation}\label{xmkq}
  J_1^0(k)^{-1}J_1(k)=I_n+o(1),\quad |k|\to\infty,\;\;k\in \overline{\mathbb{C}}^{+}.
\end{equation}
On the other hand, from (\ref{xmk}) and (\ref{xmk2}), we have
\begin{equation*}
  A_1=\frac{1}{2}(U_0+I_n), \;\; B_1=\frac{i}{2}(I_n-U_0),
\end{equation*}
which implies from $U_0=\widetilde{U}_0$ that $A_1=\widetilde{A}_1$ and $B_1=\widetilde{B}_1$, and hence $J_1^0(k)=\widetilde{J}_1^0(k)$.
Inserting $J_1^0(k)$ in (\ref{xcoa}, we have from (\ref{xmkq}) that the matrix-valued function
\begin{equation*}
  Y(k)=\left\{\begin{split}
         &[J_1^0(-k)^{-1}\widetilde{J}_1(-k)]^{-1}J_1^0(-k)^{-1}J_1(-k)=I_n+o(1), \;|k|\to\infty,\;k\in \overline{\mathbb{C}}^-,\\
          & [J_1^0(k)^{-1}\widetilde{J}_1(k)]^{-1}J_1^0(k)^{-1}J_1(k)=I_n+o(1),\;|k|\to\infty,\; k\in \overline{\mathbb{C}}^+.
        \end{split}\right.
\end{equation*}
By the Liouville theorem, we conclude that $Y(k)\equiv I_n$, and hence $J_1(k)=\widetilde{J}_1(k)$ for $k\in \overline{\mathbb{C}}^+$ from (\ref{xcoa}). The proof is complete.
\end{proof}

From the arguments in Sections 2-4 and Lemma \ref{l5.3}, we see that if the potential $V\in L_1^1(\mathbb{R}^+)$ and satisfies (\ref{2}), and the matrix $U$ in (\ref{4}) is unitary, then the scattering data $\mathcal{S}$ satisfies the assertions (I)$-$(III) presented in Section 4, and the condition (IV) presented in Theorem \ref{t5.1}.

Conversely, by Theorems \ref{t5.2}, \ref{t5.1}, \ref{t6.2} and \ref{t7.1} we see that  the data $\mathcal{S}$ satisfying the conditions (I)$-$(IV) corresponds to a unique self-adjoint potential $V\in L_1^1(\mathbb{R}^+)$ and a unique unitary matrix $U$. Therefore, we arrive at the following theorem---solvability conditions.
\begin{theorem}\label{t8.1}
For the data $\{S(k),k_j,C_j\}_{k\in\mathbb{R};j=\overline{1,N}}$ to be the scattering data of a certain boundary value problem $L(V,U)$ with a unitary matrix $U$ and a self-adjoint potential $V\in L_1^1(\mathbb{R}^+)$, it is necessary and sufficient to satisfy the conditions (I)$-$(IV).
\end{theorem}

\noindent {\bf Acknowledgments.} The authors would like to thank Professor Natalia Boundarenko for discussion and giving some corrections.
The research work was supported in part by the Natural Science Foundation of Jiangsu Province of China (BK 20141392) and  National Natural Science Foundation of China (11171152, 91538108 and 11611530682).
%%%%%%%%%%%%%%%%%%%%%%%%%%%%%%%%%%%%%

\end{document}